\documentclass[11pt]{article}
\usepackage{graphicx,amsmath,amsbsy,amssymb,amsthm,algorithmic,algorithm,cite,citesort,url}

\usepackage[usenames]{color}

\newtheorem{thm}{Theorem}[section] 
\newtheorem{lemma}{Lemma}[section] 
\newtheorem{cor}{Corollary}[section]
\newtheorem{prop}{Proposition}[section]

\newcommand{\supp}{\mathrm{supp}}

\newcommand{\argmin}{\mathop{\mathrm{arg \ min}}}

\newcommand{\complex}{\mathbb{C}}

\newcommand{\cS}{\mathcal{S}}
\newcommand{\cP}{\mathcal{P}}
\newcommand{\cR}{\mathcal{R}}

\newcommand{\bD}{\boldsymbol{D}}

\newcommand{\br}{\boldsymbol{r}}
\newcommand{\bx}{\boldsymbol{x}}

\newcommand{\balpha}{\boldsymbol{\alpha}}
\newcommand{\bA}{\boldsymbol{A}}
\newcommand{\by}{\boldsymbol{y}}

\newcommand{\be}{\boldsymbol{e}}

\newcommand{\bh}{\widetilde{\boldsymbol{v}}}

\newcommand{\bz}{\boldsymbol{z}}

\newcommand{\bv}{\boldsymbol{v}}

\newcommand{\bI}{\boldsymbol{I}}

\newcommand{\defby}{\overset{\mathrm{\scriptscriptstyle{def}}}{=}}
\newcommand{\vtil}{\widetilde{\bv}}

\newcommand{\norm}[2][2]{\left\| #2 \right\|_{#1}}
\newcommand{\norml}[2][2]{\left\| #2 \right\|_{1}}

\newlength{\imgwidth}
\title{{\bf Signal Space CoSaMP for Sparse Recovery with Redundant Dictionaries}}

\author{Mark A. Davenport$^g$, Deanna Needell$^c$, and Michael B. Wakin$^m$\thanks{
This work was partially supported by NSF grants DMS-1004718 and CCF-0830320, NSF CAREER grant CCF-1149225, and AFOSR grant FA9550-09-1-0465.}\\[3mm]
\small $^g$School of Electrical and Computer Engineering, Georgia Institute of Technology, \\[-1mm]
\small 777 Atlantic Dr. NW, Atlanta, GA 30332, USA. Email: mdav@gatech.edu\\[1mm]
\small $^c$Department of Mathematics and Computer Science, Claremont McKenna College, \\[-1mm]
\small 850 Columbia Ave., Claremont, CA 91711, USA. Email: dneedell@cmc.edu \\[1mm]
\small $^m$Department of Electrical Engineering and Computer Science, Colorado School of Mines, \\[-1mm]
\small 1500 Illinois St., Golden, CO 80401, USA. Email: mwakin@mines.edu
}

\date{August 2012 (Revised February 2013, June 2013)}

\begin{document}
\maketitle

\begin{abstract}
Compressive sensing (CS) has recently emerged as a powerful framework for acquiring sparse signals.  The bulk of the CS literature has focused on the case where the acquired signal has a sparse or compressible representation in an orthonormal basis.  In practice, however, there are many signals that cannot be sparsely represented or approximated using an orthonormal basis, but that {\em do} have sparse representations in a redundant dictionary.  Standard results in CS can sometimes be extended to handle this case provided that the dictionary is sufficiently incoherent or well-conditioned, but these approaches fail to address the case of a truly redundant or overcomplete dictionary.  In this paper we describe a variant of the iterative recovery algorithm CoSaMP for this more challenging setting.  We utilize the $\bD$-RIP, a condition on the sensing matrix analogous to the well-known restricted isometry property.  In contrast to prior work, the method and analysis are ``signal-focused''; that is, they are oriented around recovering the {\em signal} rather than its dictionary {\em coefficients}.  Under the assumption that we have a near-optimal scheme for projecting vectors in signal space onto the model family of candidate sparse signals, we provide provable recovery guarantees.  Developing a practical algorithm that can provably compute the required near-optimal projections remains a significant open problem, but we include simulation results using various heuristics that empirically exhibit superior performance to traditional recovery algorithms.
\end{abstract}

\section{Introduction}

\subsection{Overview}

Compressive sensing (CS) is a powerful new framework for signal acquisition, offering the promise that we can acquire a vector $\bx\in\complex^n$ via only $m \ll n$ linear measurements provided that $\bx$ is {\em sparse} or {\em compressible}.\footnote{When we say that a vector $\bz$ is $k$-sparse, we mean that $\norm[0]{\bz} \defby |\supp(\bz)| \leq k \ll n$.  A compressible vector is one that is well-approximated as being sparse.  We discuss compressibility in greater detail in Section~\ref{ssec:compress}.}  Specifically, CS considers the problem where we obtain measurements of the form $\by = \bA\bx + \be$, where $\bA$ is an $m\times n$ sensing matrix and $\be$ is a noise vector.  If $\bx$ is sparse or compressible and $\bA$ satisfies certain conditions, then CS provides a mechanism to recover the signal $\bx$ from the measurement vector $\by$ efficiently and robustly.

Typically, however, signals of practical interest are not themselves sparse, but rather have a sparse expansion in some dictionary $\bD$. By this we mean that there exists a sparse coefficient vector $\balpha$ such that the signal $\bx$ can be expressed as $\bx = \bD \balpha$. One could then ask the simple question: How can we account for this signal model in CS? In some cases, there is a natural way to extend the standard CS formulation---since we can write the measurements as $\by = \bA\bD\balpha + \be$ we can use standard CS techniques to first obtain an estimate $\widehat{\balpha}$ of the sparse coefficient vector.  We can then synthesize an estimate $\widehat{\bx} = \bD \widehat{\balpha}$ of the original signal. Unfortunately, this is a rather restrictive way to proceed for two main reasons: (i) the application of standard CS results to this problem will require that the matrix given by the product $\bA \bD$ satisfy certain properties that will not be satisfied for many interesting choices of $\bD$, as discussed further below, and (ii) we are not really interested in recovering $\balpha$ {\em per se}, but rather in obtaining an accurate estimate of $\bx$.  If the dictionary $\bD$ is poorly conditioned, the signal space recovery error $\norm{\bx - \widehat{\bx}}$ could be significantly smaller or larger than the coefficient space recovery error $\norm{\balpha - \widehat{\balpha}}$. It may be possible to recover $\bx$ in situations where recovering $\balpha$ is impossible, and even if we could apply standard CS results to ensure that our estimate of $\balpha$ is accurate, this would not necessarily translate into a recovery guarantee for $\bx$.

In this paper we will consider an alternative approach to this problem and develop an algorithm for which we can provide guarantees on the recovery of $\bx$ while making no direct assumptions concerning our choice of $\bD$.  Before we describe our approach, however, it will be illuminating to see precisely what goes wrong in an attempt to extend the standard CS formulation.  Towards this end, let us return to the case where $\bx$ is itself sparse (when $\bD = \bI$).  In this setting, there are many possible algorithms that have been proposed for recovering an estimate of $\bx$ from measurements of the form $\by = \bA \bx + \be$, including $\ell_1$-minimization approaches\cite{donoho2006compressed,Cande_Restricted} and greedy/iterative methods such as iterative hard thresholding (IHT)~\cite{BlumeD_Iterative}, orthogonal matching pursuit (OMP)~\cite{PatiRK_Orthogonal,MallaZ_Matching,ZhangOMP} and compressive sampling matching pursuit (CoSaMP)~\cite{NeedeT_CoSaMP}. For any of these algorithms, it can be shown (see~\cite{Cande_Restricted,BlumeD_Iterative,NeedeT_CoSaMP,DavenW_Analysis}) that $\bx$ can be accurately recovered from the measurements $\by$ if the matrix $\bA$ satisfies a condition introduced in~\cite{CandeT_Decoding} known as the restricted isometry property (RIP) with a sufficiently small constant $\delta_k$ (the precise requirement on $\delta_k$ varies from method to method).  We say that $\bA$ satisfies the RIP of order $k$ with constant  $\delta_k \in (0,1)$ if
\begin{equation}
\label{eq:ARIP} \sqrt{1-\delta_k} \le \frac{\norm{\bA \bx}}{\norm{\bx}} \le \sqrt{1+\delta_k}
\end{equation}
holds for all $\bx$ satisfying $\norm[0]{\bx} \le k$.  Importantly, if $\bA$ is generated randomly with independent and identically distributed (i.i.d.)\ entries drawn from a suitable distribution, and with a number of rows roughly proportional to the sparsity level $k$, then with high probability $\bA$ will satisfy~\eqref{eq:ARIP}~\cite{RV06:Sparse-Reconstruction,BaranDDW_Simple,mendelson2008uniform}.

We now return to the case where $\bx$ is sparse with respect to a dictionary $\bD$.  If $\bD$ is unitary (i.e., if the dictionary is an orthonormal basis), then the same arguments used to establish~\eqref{eq:ARIP} can be adapted to show that for a fixed $\bD$, if we choose a random $\bA$, then with high probability $\bA \bD$ will satisfy the RIP, i.e.,
\begin{equation}
\label{eq:ADRIP} \sqrt{1-\delta_k} \le \frac{\norm{\bA \bD \balpha}}{\norm{\balpha}} \le \sqrt{1+\delta_k}
\end{equation}
will hold for all $\balpha$ satisfying $\norm[0]{\balpha} \le k$. Thus, standard CS algorithms can be used to accurately recover $\balpha$, and because $\bD$ is unitary, the signal space recovery error $\norm{\bx - \widehat{\bx}}$ will exactly equal the coefficient space recovery error $\norm{\balpha - \widehat{\balpha}}$.

Unfortunately, this approach won't do in cases where $\bD$ is not unitary and especially in cases where $\bD$ is highly redundant/overcomplete.  For example, $\bD$ might represent the overcomplete Discrete Fourier Transform (DFT), the Undecimated Discrete Wavelet Transform, a redundant Gabor dictionary, or a union of orthonormal bases.  The challenges that we must confront when dealing with $\bD$ of this form include:
\begin{itemize}
\item Redundancy in $\bD$ will mean that in general, the representation of a vector $\bx$ in the dictionary is not unique---there may exist many possible coefficient vectors $\balpha$ that can be used to synthesize $\bx$.
\item Coherence (correlations) between the columns of $\bD$ can make it difficult to satisfy (\ref{eq:ADRIP}) with a sufficiently small constant $\delta_k$ to apply existing theoretical guarantees.  For instance, while the DFT forms an orthonormal basis, a $2 \times$ overcomplete DFT is already quite coherent, with adjacent columns satisfying $| \langle d_i, d_{i+1} \rangle | > 2/\pi > 0.63$. Since the coherence provides a bound on $\delta_k$ for all $k\ge2$, this means that $\bD$ itself cannot satisfy the RIP with a constant $\delta_{2k} < 0.63$. For the random constructions of $\bA$ typically considered in the context of CS, with high probability $\bA$ will preserve the conditioning of $\bD$ (good or bad) on each subspace of interest.  (For such $\bA$, one essentially needs $\bD$ itself to satisfy the RIP in order to expect $\bA \bD$ to satisfy the RIP.)  Thus, in the case of the $2 \times$ overcomplete DFT, we would expect the RIP constant for $\bA \bD$ to be {\em at least} roughly $0.63$---well outside the range for which any of the sparse recovery algorithms described above are known to succeed.
\item As noted above, if the dictionary $\bD$ is poorly conditioned, the signal space recovery error $\norm{\bx - \widehat{\bx}}$ could differ substantially from the coefficient space recovery error $\norm{\balpha - \widehat{\balpha}}$, further complicating any attempt to understand how well we can recover $\bx$ by appealing to results concerning the recovery of $\balpha$.
\end{itemize}
All of these problems essentially stem from the fact that extending standard CS algorithms in an attempt to recover $\balpha$ is a {\em coefficient-focused} recovery strategy. By trying to go from the measurements $\by$ all the way back to the coefficient vector $\balpha$, one encounters all the problems above due to the lack of orthogonality of the dictionary.

In contrast, in this paper we propose a {\em signal-focused} recovery strategy for CS. Our algorithm employs the model of sparsity in an arbitrary dictionary $\bD$ but directly obtains an estimate of the signal $\bx$, and we provide guarantees on the quality of this estimate in signal space. Our algorithm is a modification of CoSaMP~\cite{NeedeT_CoSaMP}, and in cases where $\bD$ is unitary, our ``Signal-Space CoSaMP'' algorithm reduces to standard CoSaMP.  However, our analysis requires comparatively weaker assumptions.  Our bounds require only that $\bA$ satisfy the $\bD$-RIP~\cite{CandeENR_Compressed} (which we explain below in Section~\ref{sec:assumptions})---this is a different and less-restrictive condition to satisfy than requiring $\bA \bD$ to satisfy the RIP.  The algorithm does, however, require the existence of a near-optimal scheme for projecting a vector $\bx$ onto the set of signals admitting a sparse representation in $\bD$.  While the fact that we require only an {\em approximate} projection is a significant relaxation of the requirements of traditional algorithms like CoSaMP (which require {\em exact} projections), showing that a practical algorithm can provably compute the required near-optimal projection remains a significant open problem.  Nevertheless, as we will see in Section~\ref{sec:sims}, various practical algorithms do lead to empirically favorable performance, suggesting that this challenge might not be insurmountable.

\subsection{Related Work}

Our work most closely relates to Blumensath's Projected Landweber Algorithm (PLA)~\cite{blumensath2011sampling}, an extension of Iterative Hard Thresholding (IHT)~\cite{BlumeD_Iterative} that operates in signal space and accounts for a union-of-subspaces signal model. In several ways, our work is a parallel of this one, except that we extend CoSaMP rather than IHT to operate in signal space. Both works assume that $\bA$ satisfies the $\bD$-RIP, and implementing both algorithms requires the ability to compute projections of vectors in the signal space onto the model family. (These requirements are described more thoroughly in Section~\ref{sec:assumptions} below.) One critical difference, however, is that our analysis allows for near-optimal projections whereas the PLA analysis does not.\footnote{Technically, the analysis of the PLA~\cite{blumensath2011sampling} allows for near-optimal projections but only with an additive error term. For sparse models, however, such projections could be made arbitrarily accurate simply by rescaling the signal before projecting. Thus, we consider the PLA to require exact projections for sparse models.} Other fine differences are noted below.

Also related are works that employ an assumption of ``analysis sparsity,'' in which a signal $\bx$ is analyzed in a dictionary $\bD$, and recovery from CS measurements is possible if $\bD^\ast \bx$ is sparse or compressible. Conventional CS algorithms such as $\ell_1$-minimization~\cite{CandeENR_Compressed,nam2011cosparse}, IHT~\cite{cevher2011alps,giryes2012greedy}, and CoSaMP~\cite{giryes2012greedy} have been adapted to account for analysis sparsity. These works are similar to ours in that they provide recovery guarantees in signal space and do not require $\bA \bD$ to satisfy the RIP. However, the assumption of analysis sparsity is in general different from the ``synthesis sparsity'' that we assume, where there exists a sparse coefficient vector $\balpha$ such that $\bx = \bD \balpha$.  For example, in the analysis case, exact sparsity implies that the analysis vector $\bD^*\bx$ is sparse, whereas in our setting exact sparsity implies the coefficient vector $\balpha$ is.  Under both of these assumptions, both the $\ell_1$-analysis method and our Signal Space CoSaMP algorithm provide recovery guarantees proportional to the norm of the noise in the measurements, $\norm{\be}$\cite{CandeENR_Compressed}.  Without exact sparsity, $\ell_1$-analysis adds an additional factor $\|\bD^*\bx - (\bD^*\bx)_k\|_1 / \sqrt{k}$, where $(\bD^*\bx)_k$ represents the $k$ largest coefficients in magnitude of $(\bD^*\bx)$.  In the synthesis sparsity setting, the analogous ``tail-term'' is less straightforward (see Section~\ref{ssec:compress} below for details).  In summary, these algorithms are intended for different signal families and potentially different dictionaries. Nevertheless, there are some similarities between our work and analysis CoSaMP (ACoSaMP)~\cite{giryes2012greedy} as we will see below.

Finally, it is worth mentioning the loose connection between our work and that in ``model-based CS''~\cite{baraniuk2010model}. In the case where $\bD$ is unitary, IHT and CoSaMP have been modified to account for structured sparsity models, in which certain sparsity patterns in the coefficient vector $\balpha$ are forbidden. This work is similar to ours in that it involves a projection onto a model set. However, the algorithm (including the projection) operates in coefficient space (not signal space) and employs a different signal model; the requisite model-based RIP is more similar to requiring that $\bA \bD$ satisfy the RIP; and extensions to non-orthogonal dictionaries are not discussed.  In fact, our work is in part inspired by our recent efforts~\cite{davenport2012compressive} to extend the ``model-based CS'' framework to a non-orthogonal dictionary in which we proposed a similar algorithm to the one considered in this paper.

\subsection{Requirements}
\label{sec:assumptions}

First, to establish notation, suppose that $\bA$ is an $m \times n$ matrix and $\bD$ is an arbitrary $n \times d$ matrix. We suppose that we observe measurements of the form $\by = \bA \bx + \be = \bA \bD \balpha + \be.$ For an index set $\Lambda \subset \{1,2,\dots,d\}$ (sometimes referred to as a {\em support set}), we let $\bD_\Lambda$ denote the $n \times |\Lambda|$ submatrix of $\bD$ corresponding to the columns indexed by $\Lambda$, and we let $\cR(\bD_\Lambda)$ denote the column span of $\bD_\Lambda$. We also use $\cP_{\Lambda}: \complex^n \rightarrow \complex^n$ to denote the orthogonal projection operator onto $\cR(\bD_\Lambda)$ and $\cP_{\Lambda^\bot}: \complex^n \rightarrow \complex^n$ to denote the orthogonal projection operator onto the orthogonal complement of $\cR(\bD_\Lambda)$.\footnote{Note that  $\cP_{\Lambda^\bot}$ does {\em not} represent the orthogonal projection operator onto $\cR(\bD_{\{1,2,\dots,d\} \backslash \Lambda})$.}

We will approach our analysis under the assumption that the matrix $\bA$ satisfies the $\bD$-RIP~\cite{CandeENR_Compressed}. Specifically, we say that $\bA$ satisfies the $\bD$-RIP of order $k$ if there exists a constant $\delta_k \in (0,1)$ such that
\begin{equation}
\label{eq:D-RIP} \sqrt{1-\delta_k} \le \frac{\norm{\bA \bD \balpha}}{\norm{\bD \balpha}} \le \sqrt{1+\delta_k}
\end{equation}
holds for all $\balpha$ satisfying $\norm[0]{\balpha} \le k$. We note that this is different from requiring that $\bA$ satisfy the RIP---although (\ref{eq:ARIP}) and (\ref{eq:D-RIP}) appear similar, the RIP requirement demands that this condition holds for vectors $\bx$ containing few nonzeros, while the $\bD$-RIP requirement demands that this condition holds for vectors $\bx$ having a sparse representation in the dictionary $\bD$. We also note that, compared to the requirement that $\bA \bD$ satisfy the RIP (\ref{eq:ADRIP}), it is relatively easy to ensure that $\bA$ satisfies the $\bD$-RIP. In particular, we have the following lemma.
\begin{lemma}[Corollary 3.1 of~\cite{davenport2012compressive}]
For any choice of $\bD$, if $\bA$ is populated with i.i.d.\ random entries from a Gaussian or subgaussian distribution, then with high probability, $\bA$ will satisfy the $\bD$-RIP of order $k$ as long as $m = O(k \log(d/k))$.
\end{lemma}
\noindent In fact, using the results of~\cite{KW:JLRIP} one can extend this result to show that given any matrix $\bA$ satisfying the traditional RIP, by applying a random sign matrix one obtains a matrix that with high probability will satisfy the $\bD$-RIP.

Next, recall that one of the key steps in the traditional CoSaMP algorithm is to project a vector in signal space onto the model family of candidate sparse signals. In the traditional setting (when $\bD$ is an orthonormal basis), this step is trivial and can be performed by simple thresholding of the entries of the coefficient vector. Our Signal Space CoSaMP algorithm (described more completely in Section~\ref{sec:alg}) involves replacing thresholding with a more general projection of vectors in the signal space onto the signal model. Specifically, for a given vector $\bz \in \complex^n$ and a given sparsity level $k$, define
$$
\Lambda_{\text{opt}}(\bz,k) := \argmin_{\Lambda: |\Lambda| = k} \norm{\bz - \cP_{\Lambda} \bz}.
$$
The support $\Lambda_{\text{opt}}(\bz,k)$---if we could compute it---could be used to generate the best $k$-sparse approximation to $\bz$; in particular, the nearest neighbor to $\bz$ among all signals that can be synthesized using $k$ columns from $\bD$ is given by $\cP_{\Lambda_{\text{opt}}(\bz,k)} \bz$. Unfortunately, computing $\Lambda_{\text{opt}}(\bz,k)$ may be difficult in general. Therefore, we allow for near-optimal projections to be used in our algorithm. For a given vector $\bz \in \complex^n$ and a given sparsity level $k$, we assume a method is available for producing an estimate of $\Lambda_{\text{opt}}(\bz,k)$, denoted $\cS_{\bD}(\bz,k)$ and having cardinality $| \cS_{\bD}(\bz,k) | = k$, that satisfies
\begin{equation}
\label{eq:approxProj}
\norm{ \cP_{\Lambda_{\text{opt}}(\bz,k)} \bz - \cP_{ \cS_{\bD}(\bz,k) } \bz } \le \min\left( \epsilon_1 \norm{ \cP_{\Lambda_{\text{opt}}(\bz,k)} \bz }, \epsilon_2 \norm{ \bz - \cP_{\Lambda_{\text{opt}}(\bz,k)} \bz }\right)
\end{equation}
for some constants $\epsilon_1, \epsilon_2 \ge 0$. Setting $\epsilon_1$ or $\epsilon_2$ equal to $0$ would lead to the requirement that $\cP_{\Lambda_{\text{opt}}(\bz,k)} \bz = \cP_{ \cS_{\bD}(\bz,k) } \bz$ exactly. Note that our metric for judging the quality of an approximation to $\Lambda_{\text{opt}}(\bz,k)$ is entirely in terms of its impact in signal space.  It might well be the case that $\cS_{\bD}(\bz,k)$ could satisfy~\eqref{eq:approxProj} while being substantially different (or even disjoint) from $\Lambda_{\text{opt}}(\bz,k)$.  Thus, while computing $\Lambda_{\text{opt}}(\bz,k)$ may be extremely challenging when $\bD$ is highly redundant, there is hope that efficiently computing an approximation that satisfies~\eqref{eq:approxProj} might still be possible.  However, determining whether this is the case remains an open problem.

It is important to note that although computing a near-optimal support estimate that satisfies the condition (\ref{eq:approxProj}) remains a challenging task in general, several important related works have run into the same problem. As we previously mentioned, the existing analysis of the PLA~\cite{blumensath2011sampling} actually requires exact computation of $\Lambda_{\text{opt}}(\bz,k)$. The analysis of ACoSaMP~\cite{giryes2012greedy} allows a near-optimal projection to be used, with a near-optimality criterion that differs slightly from ours. Simulations of ACoSaMP, however, have relied on practical (but not theoretically backed) methods for computing this projection. In Section~\ref{sec:sims}, we present simulation results for Signal Space CoSaMP using practical (but not theoretically backed) methods for computing $\cS_{\bD}(\bz,k)$.  We believe that computing provably near-optimal projections is an important topic worthy of further study, as it is really the crux of the problem in all of these settings.

\section{Algorithm and Recovery Guarantees}
\label{sec:alg}

Given noisy compressive measurements of the form $\by = \bA \bx + \be$, our Signal Space CoSaMP algorithm for recovering an estimate of the signal $\bx$ is specified in Algorithm~\ref{alg:cosampModified}.

\begin{algorithm}[t]
\caption{Signal Space CoSaMP} \label{alg:cosampModified}
\begin{algorithmic}
\STATE \textbf{input:} $\bA$, $\bD$, $\by$, $k$, stopping criterion
\STATE \textbf{initialize:} $\br = \by$, $\bx^0 = 0$, $\ell = 0$, $\Gamma=\emptyset$
\WHILE{not converged}
\STATE
\begin{tabular}{ll}
\textbf{proxy:} & $\bh = \bA^\ast \br $ \\
\textbf{identify:} & $\Omega = \cS_{\bD}(\bh,2k)$ \\
\textbf{merge:} & $T = \Omega \cup \Gamma$ \\
\textbf{update:} & $\widetilde{\bx} = \argmin_{\bz}  \norm{\by - \bA \bz} \quad \mathrm{s.t.} \quad \bz \in \cR(\bD_T)$ \\
 & $ \Gamma = \cS_{\bD}(\widetilde{\bx},k)$ \\
 & $ \bx^{\ell+1} = \cP_{ \Gamma } \widetilde{\bx} $ \\
 & $\br = \by - \bA \bx^{\ell+1}$   \\
 & $\ell = \ell+1$
\end{tabular}
\ENDWHILE
\STATE \textbf{output:} $\widehat{\bx} = \bx^{\ell}$
\end{algorithmic}
\end{algorithm}

\subsection{A Bound for the Recovery of Sparse Signals}

For signals having a sparse representation in the dictionary $\bD$, we have the following guarantee.

\begin{thm}\label{thm:exactSparse}
Suppose there exists a $k$-sparse coefficient vector $\balpha$ such that $\bx = \bD \balpha$, and suppose that $\bA$ satisfies the $\bD$-RIP of order $4k$. Then the signal estimate $\bx^{\ell+1}$ obtained after $\ell+1$ iterations of Signal Space CoSaMP satisfies
\begin{align}
\label{eq:iterativeExactSparse}
\norm{\bx - \bx^{\ell+1}} &\le C_1 \norm{\bx - \bx^{\ell}} + C_2 \norm{\be},
\end{align}
where $C_1$ and $C_2$ are constants that depend on the isometry constant $\delta_{4k}$ and on the approximation parameters $\epsilon_1$ and $\epsilon_2$. In particular,
$$
C_1 = ((2 + \epsilon_1) \delta_{4k} + \epsilon_1) (2+\epsilon_2) \sqrt{ \frac{1+\delta_{4k}}{1-\delta_{4k}}} \quad \text{and} \quad
C_2 = \left( \frac{(2+\epsilon_2)\left((2+\epsilon_1) (1 + \delta_{4k}) + 2\right)}{\sqrt{1-\delta_{4k}}} \right).
$$
\end{thm}

Our proof of Theorem~\ref{thm:exactSparse} appears in Appendix~\ref{sec:proofExactSparse} and is a modification of the original CoSaMP proof~\cite{NeedeT_CoSaMP}. Through various combinations of $\epsilon_1$, $\epsilon_2$, and $\delta_{4k}$, it is possible to ensure that $C_1 < 1$ and thus that the accuracy of Signal Space CoSaMP improves at each iteration. Taking $\epsilon_1 = \frac{1}{10}$, $\epsilon_2 = 1$, and $\delta_{4k} = 0.029$ as an example, we obtain $C_1 \le 0.5$ and $C_2 \le 12.7$. Applying the relation (\ref{eq:iterativeExactSparse}) recursively, we then conclude the following.

\begin{cor}\label{cor:exactSparse}
Suppose there exists a $k$-sparse coefficient vector $\balpha$ such that $\bx = \bD \balpha$, and suppose that $\bA$ satisfies the $\bD$-RIP of order $4k$ with $\delta_{4k} = 0.029$. Suppose that Signal Space CoSaMP is implemented using near-optimal projections with approximation parameters $\epsilon_1 = \frac{1}{10}$ and $\epsilon_2 = 1$. Then the signal estimate $\bx^{\ell}$ obtained after $\ell$ iterations of Signal Space CoSaMP satisfies
\begin{equation}
\label{eq:finalExactSparse}
\norm{\bx - \bx^{\ell}} \leq 2^{-\ell}\norm{\bx} + 25.4 \norm{\be}.
\end{equation}
\end{cor}

By taking a sufficient number of iterations $\ell$, the first term on the right hand side of (\ref{eq:finalExactSparse}) can be made arbitrarily small, and ultimately the recovery error depends only on the level of noise in the measurements.  For a precision parameter $\eta$, this shows that at most $O(\log(\norm{\bx}/\eta))$ iterations are needed to ensure that
$$
\norm{\bx - \widehat{\bx}} = O(\eta + \norm{\be}) = O(\max\left\{\eta, \norm{\be}\right\}).
$$
The cost of a single iteration of the method is dominated by the cost of the \textbf{identify} and \textbf{update} steps, where we must obtain sparse approximations to $\bh$ and $\widetilde{\bx}$, respectively.  We emphasize again that there is no known  algorithm for computing the approximation $\cS_{\bD}$ efficiently, and the ultimate computational complexity of the algorithm will depend on this choice.  However, in the absence of a better choice, one natural option for estimating $\cS_{\bD}$ is to use a greedy method such as OMP or CoSaMP (see Section~\ref{sec:sims} for experimental results using these choices).  The running time of these greedy methods on an $n\times d$ dictionary $\bD$ are $O(knd)$ or $O(nd)$, respectively~\cite{NeedeT_CoSaMP}.  Therefore, using these methods as approximations in the \textbf{identify} and \textbf{update} steps yields an overall running time of $O(knd\log(\norm{\bx}/\eta))$ or $O(nd\log(\norm{\bx}/\eta))$. For sparse signal recovery, these running times are in line with state-of-the-art bounds for traditional CS algorithms such as CoSaMP~\cite{NeedeT_CoSaMP}, except that Signal Space CoSaMP can be applied in settings where the dictionary $\bD$ is not unitary.  The error bounds of Corollary~\ref{cor:exactSparse} also match those of classical results, except that here we assume suitable accuracy of $\cS_{\bD}$.  Of course, since no efficient near-optimal projection method is known, this presents a weakness in these results, but it is one shared by all comparable results in the existing literature.

\subsection{A Discussion Concerning the Recovery of Arbitrary Signals}
\label{ssec:compress}

We can extend our analysis to account for signals $\bx$ that do not exactly have a sparse representation in the dictionary $\bD$. For the sake of illustration, we again take $\epsilon_1 = \frac{1}{10}$, $\epsilon_2 = 1$, and $\delta_{4k} = 0.029$, and we show how (\ref{eq:finalExactSparse}) can be extended.

We again assume measurements of the form $\by = \bA \bx + \be$ but allow $\bx$ to be an arbitrary signal in $\complex^n$. For any vector $\balpha_k \in \complex^d$ such that $\norm[0]{\balpha_k} \le k$, we can write
$$
\by = \bA \bx + \be = \bA\bD\balpha_k + \bA (\bx - \bD\balpha_k)  + \be.
$$
The term $\widetilde{\be} := \bA (\bx - \bD\balpha_k)  + \be$ can be viewed as noise in the measurements of the $k$-sparse signal $\bD\balpha_k$.  Then by~\eqref{eq:finalExactSparse} we have
\begin{equation}
\norm{\bD\balpha_k - \bx^{\ell}} \leq 2^{-\ell}\norm{\bD\balpha_k} + 25.4 \norm{\widetilde{\be}} \leq 2^{-\ell}\norm{\bD\balpha_k} + 25.4\left( \norm{\be} + \norm{\bA (\bx - \bD\balpha_k)}\right),
\label{eq:tail1}
\end{equation}
and so using the triangle inequality,
\begin{equation}
\norm{\bx - \bx^{\ell}} \leq 2^{-\ell}\norm{\bD\balpha_k} + 25.4 \norm{\be} + \norm{\bx - \bD\balpha_k} + 25.4 \norm{\bA (\bx - \bD\balpha_k)}.
\label{eq:tail2}
\end{equation}
One can then choose the coefficient vector $\balpha_k$ to minimize the right hand side of (\ref{eq:tail2}).

Bounds very similar to this appear in the analysis of both the PLA~\cite{blumensath2011sampling} and ACoSaMP~\cite{giryes2012greedy}.  Such results are somewhat unsatisfying since it is unclear how the term $\norm{\bA (\bx - \bD\balpha_k)}$ might behave.  Unfortunately, these bounds are difficult to improve upon when $\bD$ is not unitary. Under some additional assumptions, however, we can make further modifications to (\ref{eq:tail2}). For example, the following proposition allows us to bound $\norm{\bA (\bx - \bD\balpha_k)}$.

\begin{prop}[Proposition 3.5 of~\cite{NeedeT_CoSaMP}] \label{rudel2}
Suppose that $\bA$ satisfies the upper inequality of the RIP, i.e., that $\norm{ \bA \bx } \leq \sqrt{1 + \delta_k} \norm{ \bx }$ holds for all $\bx \in \complex^n$ with $\norm[0]{\bx} \leq k$. Then, for every signal $\bz \in \complex^n$,
\begin{equation}\label{rhs2}
\norm{ \bA \bz } \leq \sqrt{1 + \delta_k}
	\left[ \norm{ \bz } + \frac{1}{\sqrt{k}}
		\norml{\bz} \right].
\end{equation}
\end{prop}

Plugging this result in to (\ref{eq:tail2}), we have
\begin{eqnarray}
\norm{\bx - \bx^{\ell}} &\leq& 2^{-\ell}\norm{\bD\balpha_k} + 25.4 \norm{\be} + (25.4 \sqrt{1 + \delta_k}+1) \norm{\bx - \bD\balpha_k} \nonumber \\ && ~ + \frac{25.4\sqrt{1 + \delta_k}}{\sqrt{k}} \norml{\bx - \bD\balpha_k}.
\label{eq:tail3}
\end{eqnarray}
For any $\bx \in \complex^n$, one could define the model mismatch quantity
$$
\text{mismatch}(\bx) := \inf_{\balpha_k: \norm[0]{\balpha_k} \le k} \left[ \norm{ \bx - \bD \balpha_k } + \frac{1}{\sqrt{k}} \norml{\bx - \bD \balpha_k} \right].
$$
We remark that this mismatch quantity is analogous to the tail bounds in the literature for methods which do not allow for redundant dictionaries.  In particular, the $\ell_1$-norm term in the classical setting is required on account of geometric results about Gelfand widths~\cite{GG84:On-widths,Kas77:The-widths}.  If this quantity is large, then the signal is far from compressible and we are not in a setting for which our method is designed.
Plugging this definition into (\ref{eq:tail3}), we obtain
\begin{equation}
\norm{\bx - \bx^{\ell}} \leq 2^{-\ell}\norm{\bD\balpha_k} + 25.4 \norm{\be} + 26.4 \sqrt{1 + \delta_k} \cdot \text{mismatch}(\bx).
\label{eq:tail4}
\end{equation}
In some sense, one can view $\text{mismatch}(\bx)$ as the ``distance'' from $\bx$ to the set of signals that are $k$-sparse in the dictionary $\bD$, except that the actual ``distance'' being measured is a mixed $\ell_2/\ell_1$ norm. In cases where one expects this distance to be small, (\ref{eq:tail4}) guarantees that the recovery error will be small.

We close this discussion by noting that if we make the stronger assumption that $\bA\bD$ actually satisfies the RIP, we can also measure the model mismatch in the coefficient space rather than the signal space. Let $\balpha \in \complex^d$ be any vector that satisfies $\bx = \bD \balpha$. Then using a natural extension of Proposition~\ref{rudel2}, we conclude that
$$
\norm{\bA (\bx - \bD\balpha_k)} = \norm{\bA\bD(\balpha - \balpha_k)} \leq \sqrt{1+\delta_k}\left(\norm{\balpha - \balpha_k} + \frac{1}{\sqrt{k}}\norml{\balpha - \balpha_k}\right).
$$
When $\balpha$ is compressible, the recovery error (\ref{eq:tail2}) will be reasonably small.

\section{Simulations}
\label{sec:sims}

As we discussed in Section~\ref{sec:assumptions}, the main difficulty in implementing our algorithm is in computing projections of vectors in the signal space onto the model family of candidate sparse signals. One such projection is required in the {\bf identify} step in Algorithm~\ref{alg:cosampModified}; another such projection is required in the {\bf update} step. Although our theoretical analysis can accommodate near-optimal support estimates $\cS_{\bD}(\bz,k)$ that satisfy the condition (\ref{eq:approxProj}), computing even near-optimal supports can be a challenging task for many dictionaries of practical interest. In this section, we present simulation results using practical (but heuristic) methods for attempting to find near-optimal supports $\cS_{\bD}(\bz,k)$. Specifically, we find ourselves in a situation that mirrors the early days of the sparse recovery literature---we would like to identify a sparse vector that well-approximates $\bz$.  This is precisely the scenario where recovery algorithms like OMP and $\ell_1$-minimization were first proposed, so despite the lack of a theoretical guarantee, we can still apply these algorithms.  Of course, if we are leaving the solid ground of theory and entering the world of heuristics, we can also just consider applying standard algorithms like OMP, CoSaMP, and $\ell_1$-minimization algorithms directly to the CS recovery problem to see how they perform.  We will see, however, that the ``Signal Space CoSaMP'' algorithms resulting from using standard solvers for $\cS_{\bD}(\bz,k)$---even though they are not quite covered by our theory---can nevertheless outperform these classical CS reconstruction techniques.

In all simulations that follow, we set the signal length $n = 256$. We let $\bD$ be an $n \times d$ dictionary (two different dictionaries are considered below), and we construct a length-$d$ coefficient vector $\balpha$ with $k = 8$ nonzero entries chosen as i.i.d.\ Gaussian random variables. We set $\bx = \bD \balpha$, construct $\bA$ as a random $m \times n$ matrix with i.i.d.\ Gaussian entries, and collect noiseless measurements $\by = \bA \bx$. After reconstructing an estimate of $\bx$, we declare this recovery to be perfect if the SNR of the recovered signal estimate is above 100~dB. All of our simulations were performed via a MATLAB software package that we have made available for download at \url{http://users.ece.gatech.edu/~mdavenport/software}.

\subsection{Renormalized Orthogonal Dictionary}

As an instructive warm-up, we begin with one almost trivial example where the optimal projection can be computed exactly: we construct $\bD$ by taking an orthobasis and renormalizing its columns while maintaining their orthogonality. To compute $\Lambda_{\text{opt}}(\bz,k)$ with such a dictionary, one merely computes $\bD^\ast \bz$, divides this vector elementwise by the column norms of $\bD$, and sets $\Lambda$ equal to the positions of the $k$ largest entries.

For the sake of demonstration, we set $\bD$ equal to the $n \times n$ identity matrix, but we then rescale its first $n/2$ diagonal entries to equal $100$ instead of $1$. We construct sparse coefficient vectors $\balpha$ as described above, with supports chosen uniformly at random. As a function of the number of measurements $m$, we plot in Figure~\ref{fig:nnd} the percent of trials in which Signal Space CoSaMP recovers $\bx$ exactly. We see that with roughly 50 or more measurements, the recovery is perfect in virtually all trials. In contrast, this figure also shows the performance of traditional CoSaMP using the combined dictionary $\bA \bD$ to first recover $\balpha$. Because of the non-normalized columns in $\bD$, CoSaMP almost never recovers the correct signal; in fact its support estimates almost always are contained in the set $\{1,2,\dots,n/2\}$. Of course, if presented with this problem in practice one would naturally want to modify traditional CoSaMP to account for the various column norms in $\bD$; the point here is merely that our algorithm gives a principled way to make this (trivial) modification.

\begin{figure}
   \centering
   \includegraphics[width=2.8in]{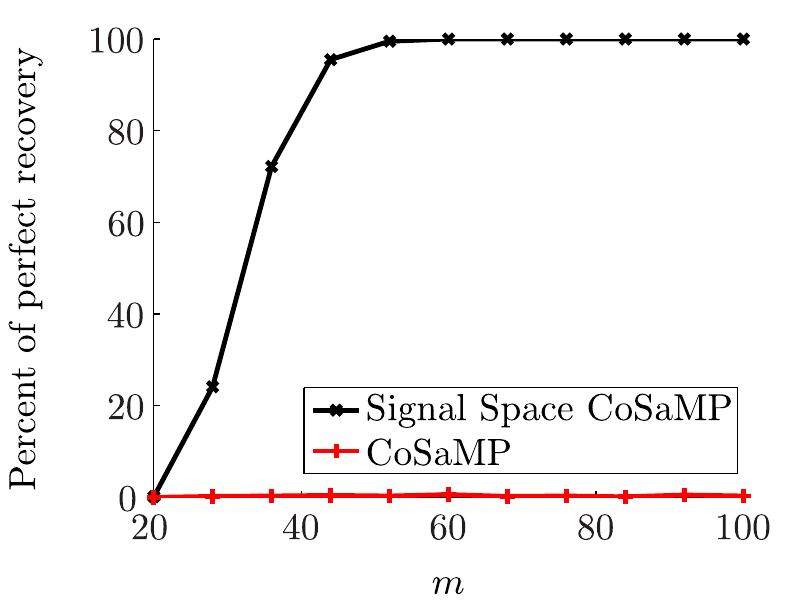}
   \caption{\small \sl Performance in recovering signals having a $k=8$ sparse representation in a dictionary $\bD$ with orthogonal, but not normalized, columns. The plot shows, for various numbers of measurements $m$, the percent of trials in which each algorithm recovers the signal exactly. Signal Space CoSaMP (in which we can compute optimal projections) outperforms an unmodified CoSaMP algorithm.
   \label{fig:nnd}}
\end{figure}

\subsection{Overcomplete DFT Dictionary}

As a second example, we set $d = 4n$ and let $\bD$ be a $4\times$ overcomplete DFT dictionary. In this dictionary, neighboring columns are highly coherent, while distant columns are not. We consider two scenarios: one in which the $k=8$ nonzero entries of $\balpha$ are randomly positioned but well-separated (with a minimum spacing of $8$ zeros in between any pair of nonzero entries), and one in which the $k=8$ nonzero entries all cluster together in a single, randomly-positioned block. Because of the nature of the columns in $\bD$, we see that many recovery algorithms perform differently in these two scenarios.

\subsubsection{Well-separated coefficients}
\label{subsec:sep}

Figure~\ref{fig:overcomplete}(a) plots the performance of six different recovery algorithms for the scenario where the nonzero entries of $\balpha$ are well-separated. Two of these algorithms are the traditional OMP and CoSaMP algorithms from CS, each using the combined dictionary $\bA \bD$ to first recover $\balpha$. We actually see that OMP performs substantially better than CoSaMP in this scenario, apparently because it can select one coefficient at a time and is less affected by the coherence of $\bD$. It is somewhat remarkable that OMP succeeds at all, given that $\bA \bD$ will not satisfy the RIP and we are not aware of any existing theory that would guarantee the performance of OMP in this scenario.

\begin{figure}
   \centering
   \begin{tabular}{cc}
   \includegraphics[width=3in]{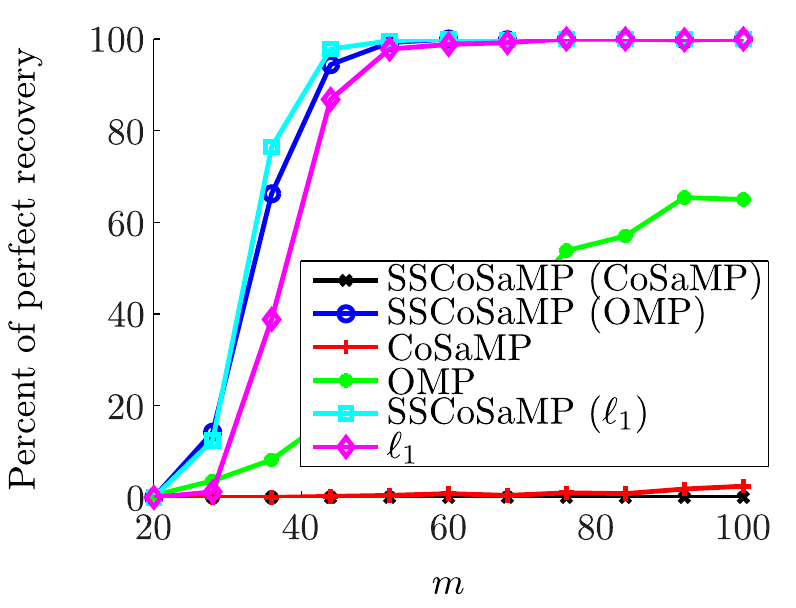} & \includegraphics[width=3in]{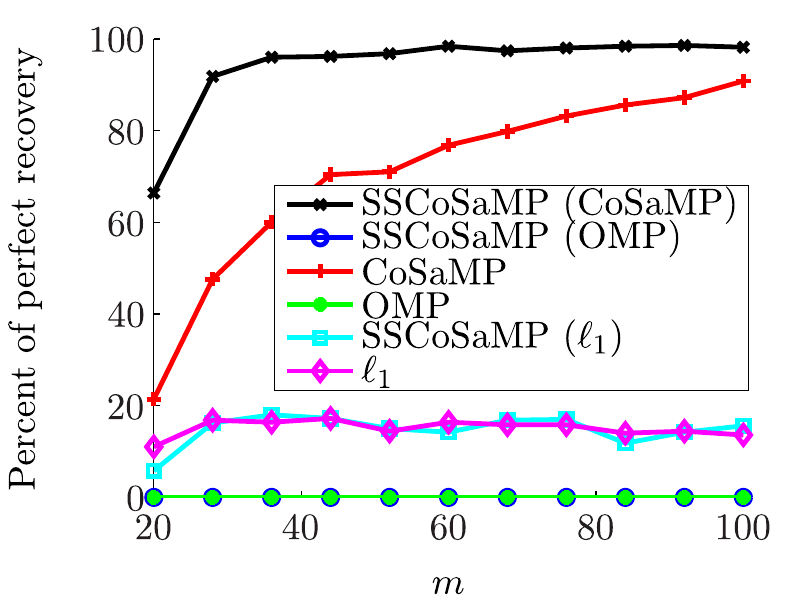} \\
   \hspace*{9mm} (a) & \hspace*{9mm} (b)
   \end{tabular}
   \caption{\small \sl Performance in recovering signals having a $k=8$ sparse representation in a 4$\times$ overcomplete DFT dictionary. Two scenarios are shown: (a)~one in which the $k=8$ nonzero entries of $\balpha$ are randomly positioned but well-separated, and (b)~one in which the $k=8$ nonzero entries all cluster together in a single, randomly-positioned block. Algorithms involving OMP and $\ell_1$-minimization perform well in the former scenario; algorithms involving CoSaMP perform well in the latter. In general, the Signal Space CoSaMP variants outperform the corresponding traditional CS algorithm.
   \label{fig:overcomplete}}
\end{figure}

We also show in Figure~\ref{fig:overcomplete}(a) two variants of Signal Space CoSaMP: one in which OMP is used for computing $\cS_{\bD}(\bz,k)$ (labeled ``SSCoSaMP (OMP)''), and one in which CoSaMP is used for computing $\cS_{\bD}(\bz,k)$ (labeled ``SSCoSaMP (CoSaMP)''). That is, these algorithms actually use OMP or CoSaMP as an inner loop inside of Signal Space CoSaMP to  find a sparse solution to the equation $\bz = \bD \balpha$. In this scenario, we see that the performance of SSCoSaMP (OMP) is substantially better than OMP, while the performance of SSCoSaMP (CoSaMP) is poor. We believe that this happens for the same reason that traditional OMP outperforms traditional CoSaMP. In general, we have found that when OMP performs well, SSCoSaMP (OMP) may perform even better, and when CoSaMP performs poorly, SSCoSaMP (CoSaMP) may still perform poorly.

Figure~\ref{fig:overcomplete}(a) also shows the performance of two algorithms that involve convex optimization for sparse regularization. One, labeled ``$\ell_1$,'' uses an $\ell_1$-minimization approach~\cite{SPGL1} to find a sparse coefficient vector $\balpha'$ subject to the constraint that $\by = \bA \bD \balpha'$. This algorithm actually outperforms traditional OMP in this scenario. The other, labeled ``SSCoSaMP ($\ell_1$),'' is a variant of Signal Space CoSaMP in which $\ell_1$-minimization is used for computing $\cS_{\bD}(\bz,k)$.\footnote{We are not unaware of the irony of using $\ell_1$-minimization inside of a greedy algorithm.} Specifically, to compute $\cS_{\bD}(\bz,k)$, we search for the vector $\balpha'$ having the smallest $\ell_1$ norm subject to the constraint that $\bz = \bD \balpha'$, and we then choose the support that contains the $k$ largest entries of this vector. Remarkably, this algorithm performs best of all. We believe that this is likely due to the fact that, for the overcomplete DFT dictionary, $\ell_1$-minimization is known to be capable of finding $\Lambda_{\text{opt}}(\bz,k)$ exactly when $\bz = \cP_{\Lambda_{\text{opt}}(\bz,k)} \bz$ and the entries of $\Lambda_{\text{opt}}(\bz,k)$ are sufficiently well-separated~\cite{Superres}. While we do not guarantee that this condition will be met within every iteration of Signal Space CoSaMP, the fact that the original coefficient vector $\balpha$ has well-separated coefficients seems to be intimately related to the success of $\ell_1$ and SSCoSaMP ($\ell_1$) here.

We note that all of the above algorithms involve a step where a least-squares problem must be solved on an estimated support set.\footnote{We use this for debiasing after running $\ell_1$.} This might seem somewhat contrary to our signal-focused approach, since solving this least-squares problem essentially involves recovering a set of coefficients $\balpha$ and then computing $\widetilde{\bx} = \bD_T \balpha$.   However, recall that at this point in the algorithm we are solving an over-determined system, so there is no significant difference between solving for $\balpha$ versus $\widetilde{\bx}$.  The problem with coefficient-focused strategies and analysis is that they rely on identifying the subset $T$ containing the ``correct'' subset of coefficients, whereas for us, $T$ could be very different from the ``correct'' subset, and the vector $\balpha$ has no particular significance---all that matters is the vector $\widetilde{\bx}$ that this step ultimately synthesizes.  Nevertheless, it is worth noting that when a dictionary $\bD$ is highly coherent, it can be numerically challenging to solve this problem, as the resulting submatrix can be very poorly conditioned. Following our discussion in~\cite{davenport2012compressive}, we employ {\em Tikhonov regularization}~\cite{Phill_technique,Tikho_Solution,TikhoA_Solutions,Hanse_Regularization} and solve a norm-constrained least-squares problem to improve its conditioning. For this we must provide our algorithms with an upper bound on the norm of the sparse coefficient vector. In the simulations above, we have selected this bound to be $10 \times$ the true norm of the original $\balpha$. The selection of this bound does not have a substantial impact on the performance of OMP or SSCoSaMP (OMP), but we have noticed that CoSaMP and SSCoSaMP (CoSaMP) perform somewhat better when the true norm $\norm{\balpha}$ is provided as an oracle.

\subsubsection{Clustered coefficients}
\label{subsec:cluster}

Figure~\ref{fig:overcomplete}(b) plots the performance of the same six recovery algorithms for the scenario where the nonzero entries of $\balpha$ are clustered into a single block. Although one could of course employ a block-sparse recovery algorithm in this scenario, our intent is more to study the impact that neighboring active atoms have on the algorithms above.

In this scenario, between the traditional greedy algorithms, CoSaMP now outperforms OMP, apparently because it is designed to select multiple indices at each step and will not be as affected by the coherence of neighboring active columns in $\bD$. We also see that the performance of SSCoSaMP (CoSaMP) is somewhat better than CoSaMP, while the performance of SSCoSaMP (OMP) is poor. We believe that this happens for the same reason that traditional CoSaMP outperforms traditional OMP. In general, we have found that when CoSaMP performs well, SSCoSaMP (CoSaMP) may perform even better, and when OMP performs poorly, SSCoSaMP (OMP) may still perform poorly.

In terms of our condition for perfect recovery (estimating $\bx$ to within an SNR of 100~dB or more), neither of the algorithms that involve $\ell_1$-minimization perform well in this scenario. However, we do note that both $\ell_1$ and SSCoSaMP ($\ell_1$) do frequently recover an estimate of $\bx$ with an SNR of 50~dB or more, though still not quite as frequently as SSCoSaMP (CoSaMP) does.

In these simulations, we again use Tikhonov regularization with a norm bound that is $10 \times$ the true norm of the original $\balpha$. However, we have not found that changing this norm bound has a significant impact on CoSaMP or SSCoSaMP (CoSaMP) in this scenario. We also note that in this scenario we found it beneficial to run CoSaMP and the three Signal Space CoSaMP methods for a few more iterations than in the case of well-separated coefficients; convergence to exactly the correct support can be slow in this case where multiple neighboring atoms in a coherent dictionary are active.

\subsubsection{Additional investigations}

We close with some final remarks concerning additional investigations. First, our simulations above have tested three heuristic methods for attempting to find near-optimal supports $\cS_{\bD}(\bz,k)$, and we have evaluated the performance of these methods based on the ultimate success or failure of SSCoSaMP in recovering the signal. In some problems of much smaller dimension (where we could use an exhaustive search to find $\Lambda_{\text{opt}}(\bz,k)$) we monitored the performance of CoSaMP, OMP, and $\ell_1$-minimization for computing $\cS_{\bD}(\bz,k)$ in terms of the effective $\epsilon_1$ and $\epsilon_2$ values they attained according to the metric in~(\ref{eq:approxProj}). For scenarios where the nonzero entries of $\balpha$ were well-separated, we observed typical $\epsilon_1$ and $\epsilon_2$ values for OMP and $\ell_1$-minimization on the order of $1$ or less. For CoSaMP, these values were larger by one or two orders of magnitude, as might be expected based on the signal recovery results presented in Section~\ref{subsec:sep}.\footnote{The occasional exception in some of these simulations occurred when it happened that $\norm{ \bz - \cP_{\Lambda_{\text{opt}}(\bz,k)} \bz } \approx 0$ but $\norm{ \cP_{\Lambda_{\text{opt}}(\bz,k)} \bz - \cP_{ \cS_{\bD}(\bz,k) } \bz }$ was not correspondingly small, and so the effective $\epsilon_2$ value was large or infinite.} For scenarios where the nonzero entries of $\balpha$ were clustered, the $\epsilon_2$ values for OMP and $\ell_1$-minimization increased by about one order of magnitude, but the $\epsilon_1$ and $\epsilon_2$ values for CoSaMP did not change significantly. The primary reason for this, despite the superior signal recovery performance of SSCoSaMP (CoSaMP) in Section~\ref{subsec:cluster}, appears to be that even when the nonzero entries of $\balpha$ are clustered, the support of the optimal approximation of $\bh$ in the \textbf{identify} step will not necessarily be clustered, and so CoSaMP will struggle to accurately identify this support.

Second, we remark that our simulations in Sections~\ref{subsec:sep} and~\ref{subsec:cluster} have tested two extremes: one scenario in which the nonzero entries of $\balpha$ were well-separated, and one scenario in which the nonzero entries clustered together in a single block. Among the heuristic methods that we have used for attempting to find near-optimal supports $\cS_{\bD}(\bz,k)$, the question of which method performs best has been shown to depend on the sparsity pattern of $\balpha$. Although we do not present detailed results here, we have also tested these same algorithms using a hybrid sparsity model for $\balpha$ in which half of the nonzero entries are well-separated while the other half are clustered. As one might expect based on the discussions above, all three of the SSCoSaMP methods struggle in this scenario (as do the three standard CS methods). This is yet another reminder that more work is needed to understand what techniques are appropriate for approximating $\cS_{\bD}(\bz,k)$ and how to optimize these techniques depending on what is known about the signal's sparsity pattern.

\appendix

\section{Proof of Theorem~\ref{thm:exactSparse}}
\label{sec:proofExactSparse}

The proof of Theorem~\ref{thm:exactSparse} requires four main lemmas, which are listed below and proved in Sections~\ref{sec:proofStatement1}--\ref{sec:proofStatement4}.  In the lemmas below, $\bv = \bx - \bx^{\ell}$ denotes the recovery error in signal space after $\ell$ iterations.

\begin{lemma}[Identify]
\label{lem:statement1}
$\norm{\cP_{\Omega^\bot} \bv} \le ((2 + \epsilon_1) \delta_{4k} + \epsilon_1) \norm{\bv} + (2+\epsilon_1) \sqrt{1 + \delta_{4k}} \norm{\be}$.
\end{lemma}

\begin{lemma}[Merge]
\label{lem:statement2}
$\norm{\cP_{T^\bot} \bx} \le \norm{\cP_{\Omega^\bot} \bv}$.
\end{lemma}

\begin{lemma}[Update]
\label{lem:statement3}
$\norm{\bx - \widetilde{\bx}} \le \sqrt{ \frac{1+\delta_{4k}}{1-\delta_{4k}}} \norm{\cP_{T^\bot} \bx} + \frac{2}{\sqrt{1-\delta_{4k}}} \norm{\be}$.
\end{lemma}

\begin{lemma}[Estimate]
\label{lem:statement4}
$\norm{\bx - \bx^{\ell+1}} \le (2+\epsilon_2) \norm{\bx - \widetilde{\bx}}$.
\end{lemma}

Combining all four statements above, we have
\begin{align*}
\norm{\bx - \bx^{\ell+1}} &\le  (2+\epsilon_2) \norm{\bx - \widetilde{\bx}}\\
&\le (2+\epsilon_2) \sqrt{ \frac{1+\delta_{4k}}{1-\delta_{4k}}} \norm{\cP_{T^\bot} \bx} + \frac{4+2\epsilon_2}{\sqrt{1-\delta_{4k}}} \norm{\be} \\
&\le (2+\epsilon_2) \sqrt{ \frac{1+\delta_{4k}}{1-\delta_{4k}}} \norm{\cP_{\Omega^\bot} \bv} + \frac{4+2\epsilon_2}{\sqrt{1-\delta_{4k}}} \norm{\be} \\
&\le  (2+\epsilon_2) \sqrt{ \frac{1+\delta_{4k}}{1-\delta_{4k}}} ((2 + \epsilon_1) \delta_{4k} + \epsilon_1) \norm{\bv}  \\
& ~~~~~ +  (2+\epsilon_2) \sqrt{ \frac{1+\delta_{4k}}{1-\delta_{4k}}} (2+\epsilon_1) \sqrt{1 + \delta_{4k}} \norm{\be} + \frac{4+2\epsilon_2}{\sqrt{1-\delta_{4k}}} \norm{\be} \\
&= ((2 + \epsilon_1) \delta_{4k} + \epsilon_1) (2+\epsilon_2) \sqrt{ \frac{1+\delta_{4k}}{1-\delta_{4k}}} \norm{\bx - \bx^{\ell}} + \left( \frac{(2+\epsilon_2)\left((2+\epsilon_1) (1 + \delta_{4k}) + 2\right)}{\sqrt{1-\delta_{4k}}} \right) \norm{\be}.
\end{align*}
This completes the proof of Theorem~\ref{thm:exactSparse}.

\subsection{Proof of Lemma~\ref{lem:statement1}}
\label{sec:proofStatement1}

In order to prove the four main lemmas, we require two supplemental lemmas, the first of which is a direct consequence of the $\bD$-RIP.

\begin{lemma}[Consequence of $\bD$-RIP]\label{lem:opnorm}
For any index set $B$ and any vector $\bz \in \complex^n$,
$$
\norm{\cP_{B}\bA^*\bA\cP_{B}\bz - \cP_{B}\bz} \le \delta_{|B|}\norm{\bz}.
$$
\end{lemma}
\begin{proof}
We have
\begin{align*}
\delta_{|B|} & \geq \sup_{\bx:~ |B|-\mathrm{sparse~in}~\bD} \frac{| \norm{\bA \bx}^2 - \norm{\bx}^2 |}{\norm{\bx}^2} \\
& \geq \sup_{\bx} \frac{| \norm{\bA \cP_{B} \bx}^2 - \norm{ \cP_{B} \bx}^2 |}{\norm{ \cP_{B} \bx}^2} \\
& \geq \sup_{\bx} \frac{| \norm{\bA \cP_{B} \bx}^2 - \norm{ \cP_{B} \bx}^2 |}{\norm{ \bx}^2} \\
&= \sup_{\norm{\bx}=1} | \norm{\bA\cP_{B}\bx}^2 - \norm{\cP_{B}\bx}^2| \allowdisplaybreaks \\
&= \sup_{\norm{\bx}=1} |\langle \cP_{B}^*\bA^*\bA\cP_{B}\bx - \cP_{B}^*\cP_{B}\bx, \bx \rangle|\\
&= \sup_{\norm{\bx}=1} |\langle \cP_{B}\bA^*\bA\cP_{B}\bx - \cP_{B}\bx, \bx \rangle|\\
&= \norm{\cP_{B}\bA^*\bA\cP_{B} - \cP_{B}},
\end{align*}
where the third line follows because $\norm{\cP_{B} \bx} \le \norm{\bx}$ for all $\bx$, and the last line follows from 
the fact that $\cP_{B}\bA^*\bA\cP_{B} - \cP_{B}$ is self-adjoint.
\end{proof}

We'll also utilize an elementary fact about orthogonal projections.

\begin{lemma}\label{lem:nestedproj}
For any pair of index sets $A, B$ with $A \subset B$, $\cP_A = \cP_A \cP_B$.
\end{lemma}

Now, to make the notation simpler, note that $\vtil = \bA^* \bA \bv + \bA^* \be$ and that $\Omega = \cS_{\bD}(\vtil,2k)$. Let $\Omega^\ast = \Lambda_{\text{opt}}(\vtil,2k)$ denote the optimal support of size $2k$ for approximating $\vtil$, and set $R = \cS_{\bD}(\bv, 2k)$. Using this notation we have
\begin{align}
\norm{\cP_{\Omega^\bot} \bv} & = \norm{ \bv - \cP_{\Omega} \bv } \nonumber \\
& \le \norm{\bv - \cP_{\Omega} \vtil} \nonumber \\
& = \norm{(\bv - \cP_{R \cup \Omega^*} \vtil) + (\cP_{R \cup \Omega^*} \vtil - \cP_{\Omega^*} \vtil) + (\cP_{\Omega^*} \vtil - \cP_{\Omega} \vtil)} \nonumber \\
& \le \norm{\bv - \cP_{R \cup \Omega^*} \vtil} + \norm{ \cP_{R \cup \Omega^*} \vtil - \cP_{\Omega^*} \vtil} + \norm{ \cP_{\Omega^*} \vtil - \cP_{\Omega} \vtil } \nonumber \\
& \le \norm{\bv - \cP_{R \cup \Omega^*} \bA^* \bA \bv} + \norm{\cP_{R \cup \Omega^*} \bA^* \be} + \norm{ \cP_{R \cup \Omega^*} \vtil - \cP_{\Omega^*} \vtil} + \norm{ \cP_{\Omega^*} \vtil - \cP_{\Omega} \vtil } , \label{eq:approxnoisylab1}
\end{align}
where the second line follows from the fact that $\cP_{\Omega} \bv$ is the nearest neighbor to $\bv$ among all vectors in $\cR(\bD_\Omega)$ and the fourth and fifth lines use the triangle inequality.

Below, we will provide bounds on the first and second terms appearing in~\eqref{eq:approxnoisylab1}. To deal with the third term in~\eqref{eq:approxnoisylab1}, note that for any $\Pi$ which is a subset of $R \cup \Omega^*$, we can write
\begin{align*}
\vtil - \cP_{\Pi} \vtil & = (\vtil - \cP_{R \cup \Omega^*} \vtil) + (\cP_{R \cup \Omega^*} \vtil - \cP_{\Pi} \vtil),
\end{align*}
where $\vtil - \cP_{R \cup \Omega^*} \vtil$ is orthogonal to $\cR(\bD_{R \cup \Omega^*})$, and $\cP_{R \cup \Omega^*} \vtil - \cP_{\Pi} \vtil$ is contained in $\cR(\bD_{R \cup \Omega^*})$. Thus we can write
\begin{align*}
\norm{\vtil - \cP_{\Pi} \vtil}^2 & = \norm{\vtil - \cP_{R \cup \Omega^*} \vtil}^2 + \norm{\cP_{R \cup \Omega^*} \vtil - \cP_{\Pi} \vtil}^2.
\end{align*}
Recall that over all index sets $\Pi$ with $|\Pi| = 2k$, $\norm{\vtil - \cP_{\Pi} \vtil}$ is minimized by choosing $\Pi = \Omega^*$. Thus, over all $\Pi$ which are subsets of $R \cup \Omega^*$ with $|\Pi| = 2k$, $\norm{\cP_{R \cup \Omega^*} \vtil - \cP_{\Pi} \vtil}^2$ must be minimized by choosing $\Pi = \Omega^*$. In particular, we have the first inequality below:
\begin{align}
\norm{ \cP_{R \cup \Omega^*} \vtil - \cP_{\Omega^*} \vtil} &\le \norm{ \cP_{R \cup \Omega^*} \vtil - \cP_{R} \vtil} \nonumber \\
&= \norm{ \cP_{R \cup \Omega^*} \vtil - \cP_{R} \cP_{R \cup \Omega^*} \vtil} \nonumber \\
& \le \norm{ \cP_{R \cup \Omega^*} \vtil - \cP_{R} (\bv + \bA^* \be)} \nonumber \\
& = \norm{ (\cP_{R \cup \Omega^*} \bA^* \bA \bv - \bv) + (\cP_{R \cup \Omega^*} \bA^* \be - \cP_{R} \bA^* \be) } \nonumber \\
& \le \norm{ \cP_{R \cup \Omega^*} \bA^* \bA \bv - \bv} + \norm{\cP_{R \cup \Omega^*} \bA^* \be - \cP_{R} \bA^* \be} \nonumber \\
& \le \norm{ \cP_{R \cup \Omega^*} \bA^* \bA \bv - \bv} + \norm{(I - \cP_R) \cP_{R \cup \Omega^*} \bA^* \be} \nonumber \\
& \le \norm{ \cP_{R \cup \Omega^*} \bA^* \bA \bv - \bv} + \norm{\cP_{R \cup \Omega^*} \bA^* \be}. \label{eq:approxnoisylab3}
\end{align}
The second line above uses Lemma~\ref{lem:nestedproj}, the third line follows from the fact that $\cP_{R} \cP_{R \cup \Omega^*} \vtil$ must be the nearest neighbor to $\cP_{R \cup \Omega^*} \vtil$ among all vectors in $\cR(\bD_R)$,  the fourth line uses the fact that $\cP_{R} \bv = \bv$ because $R = \cS_{\bD}(\bv, 2k)$ and both $\bx$ and $\bx^{\ell}$ are $k$-sparse in $\bD$, the fifth line uses the triangle inequality, the sixth line uses Lemma~\ref{lem:nestedproj}, and the seventh line follows from the fact that $(I - \cP_R)$ is an orthogonal projection and hence has norm bounded by 1.

To deal with the fourth term in~\eqref{eq:approxnoisylab1}, note that from the definition of $\Omega^*$ and from \eqref{eq:approxProj}, we have the first inequality below:
\begin{align}
\norm{ \cP_{\Omega^*} \vtil - \cP_{\Omega} \vtil }
&\le \epsilon_1 \norm{ \cP_{\Omega^*} \vtil } \nonumber \\
&= \epsilon_1 \norm{ \cP_{\Omega^*} (\bA^* \bA \bv + \bA^* \be) } \nonumber \\
&\le \epsilon_1 \norm{ \cP_{\Omega^*} \bA^* \bA \bv} + \epsilon_1 \norm{ \cP_{\Omega^*} \bA^* \be } \nonumber \\
&= \epsilon_1 \norm{ \cP_{\Omega^*} \cP_{R \cup \Omega^*} \bA^* \bA  \bv} + \epsilon_1 \norm{ \cP_{\Omega^*} \cP_{R \cup \Omega^*} \bA^* \be } \nonumber \\
&\le \epsilon_1 \norm{ \cP_{R \cup \Omega^*} \bA^* \bA  \bv} + \epsilon_1 \norm{ \cP_{R \cup \Omega^*} \bA^* \be } \nonumber \\
&\le \epsilon_1 \norm{ \bv} + \epsilon_1 \norm{ \cP_{R \cup \Omega^*} \bA^* \bA  \bv - \bv} + \epsilon_1 \norm{ \cP_{R \cup \Omega^*} \bA^* \be },  \label{eq:approxnoisylab2}
\end{align}
The third line above uses the triangle inequality, the fourth line uses Lemma~\ref{lem:nestedproj}, the fifth line uses the fact $\cP_{\Omega^*}$ is an orthogonal projection and hence has norm bounded by 1, and the sixth line uses the triangle inequality.

Combining~\eqref{eq:approxnoisylab1}, \eqref{eq:approxnoisylab3}, and \eqref{eq:approxnoisylab2} we see that
$$
\norm{\cP_{\Omega^\bot} \bv} \le (2+\epsilon_1) \norm{\cP_{R \cup \Omega^*} \bA^* \bA \bv - \bv} + (2+\epsilon_1) \norm{\cP_{R \cup \Omega^*} \bA^* \be} + \epsilon_1 \norm{ \bv}.
$$
Since $\bv \in \cR(\bD_R)$, it follows that $\bv \in \cR(\bD_{R \cup \Omega^*})$, and so
\begin{align*}
\norm{ \cP_{R \cup \Omega^*} \bA^* \bA \bv - \bv } =  \norm{ \cP_{R \cup \Omega^*} \bA^* \bA \cP_{R \cup \Omega^*} \bv - \cP_{R \cup \Omega^*} \bv } \le \delta_{4k} \norm{\bv},
\end{align*}
where we have used Lemma~\ref{lem:opnorm} to get the inequality above. In addition, we know that the operator norm of $\cP_{R \cup \Omega^*} \bA^*$ satisfies
$$
\norm{\cP_{R \cup \Omega^*} \bA^*} = \norm{(\cP_{R \cup \Omega^*} \bA^*)^\ast} = \norm{\bA \cP_{R \cup \Omega^*}} \le \sqrt{1 + \delta_{4k}},
$$
which follows from the $\bD$-RIP.  Specifically, for any $\bx$,
$$
\frac{\norm{\bA \cP_{R \cup \Omega^*} \bx}}{\norm{\bx}} \le \frac{\norm{\bA \cP_{R \cup \Omega^*} \bx}}{\norm{\cP_{R \cup \Omega^*} \bx}} \le \sqrt{1 + \delta_{4k}}.
$$
Putting all of this together, we have
$$
\norm{\cP_{\Omega^\bot} \bv} \le ((2 + \epsilon_1) \delta_{4k} + \epsilon_1) \norm{\bv} + (2+\epsilon_1) \sqrt{1 + \delta_{4k}} \norm{\be}.
$$

\subsection{Proof of Lemma~\ref{lem:statement2}}
\label{sec:proofStatement2}

First note that by the definition of $T$, $\bx^{\ell} \in \cR(\bD_T)$, and hence $\cP_{T^{\bot}} \bx^{\ell} = 0$.  Thus we can write,
$$
\norm{\cP_{T^{\bot}} \bx} = \norm{\cP_{T^{\bot}} (\bx - \bx^{\ell}) } = \norm{\cP_{T^{\bot}} \bv }.
$$
Finally, since $\Omega \subseteq T$, we have that
$$
\norm{\cP_{T^{\bot}} \bv } \le \norm{\cP_{\Omega^{\bot}} \bv }.
$$

\subsection{Proof of Lemma~\ref{lem:statement3}}
\label{sec:proofStatement3}

To begin, we note that $\bx - \widetilde{\bx}$ has a $4k$-sparse representation in $\bD$, thus, applying the $\bD$-RIP (of order $4k$) we have
$$
\norm{\bx - \widetilde{\bx}} \le \frac{ \norm{ \bA \bx - \bA \widetilde{\bx} }}{\sqrt{1-\delta_{4k}}}.
$$
By construction,
$$
\norm{\bA \bx - \bA \widetilde{\bx} + \be} \le \norm{\bA \bx - \bA \bz + \be}
$$
for any $\bz \in \cR(\bD_T)$, in particular for $\bz = \cP_{T} \bx$. Thus,
\begin{align*}
\norm{\bx - \widetilde{\bx}} & \le \frac{ \norm{ \bA \bx - \bA \widetilde{\bx} }}{\sqrt{1-\delta_{4k}}} \\
& \le \frac{ \norm{ \bA \bx - \bA \widetilde{\bx} + \be } + \norm{\be} }{\sqrt{1-\delta_{4k}}} \\
& \le \frac{ \norm{\bA \bx - \bA \cP_{T} \bx + \be} + \norm{\be} }{\sqrt{1-\delta_{4k}}} \\
& \le \frac{ \norm{\bA \bx - \bA \cP_{T} \bx} + 2 \norm{\be} }{\sqrt{1-\delta_{4k}}}
\end{align*}
where the second and fourth lines use the triangle inequality. By applying the $\bD$-RIP we obtain
$$
\norm{ \bA \bx - \bA \cP_{T} \bx } \le \sqrt{1+\delta_{4k}} \norm{\bx - \cP_{T} \bx} = \sqrt{1+\delta_{4k}} \norm{\cP_{T^\bot} \bx}.
$$
Combining all of this,
$$
\norm{\bx - \widetilde{\bx}} \le \frac{ \sqrt{1+\delta_{4k}} \norm{\cP_{T^\bot} \bx} + 2 \norm{\be} }{\sqrt{1-\delta_{4k}}}.
$$

\subsection{Proof of Lemma~\ref{lem:statement4}}
\label{sec:proofStatement4}

Using the triangle inequality, we have
$$
\norm{\bx - \bx^{\ell+1}} = \norm{\bx - \widetilde{\bx} + \widetilde{\bx} - \bx^{\ell+1}} \le \norm{\bx - \widetilde{\bx}} + \norm{\widetilde{\bx} - \bx^{\ell+1}}.
$$
Recall that $\Gamma = \cS_{\bD}(\widetilde{\bx},k)$ and $\bx^{\ell+1} = \cP_{ \Gamma } \widetilde{\bx}$. Let $\Gamma^* = \Lambda_{\text{opt}}(\widetilde{\bx},k)$ denote the optimal support of size $k$ for approximating $\widetilde{\bx}$. Then we can write
\begin{align*}
\norm{\widetilde{\bx} - \bx^{\ell+1}} &\le \norm{\widetilde{\bx} - \cP_{\Gamma^*} \widetilde{\bx}} + \norm{\cP_{\Gamma^*} \widetilde{\bx} - \cP_{\Gamma} \widetilde{\bx}} \\
&\le  \norm{\widetilde{\bx} - \cP_{\Gamma^*} \widetilde{\bx}} + \epsilon_2 \norm{\widetilde{\bx} - \cP_{\Gamma^*} \widetilde{\bx}},
\end{align*}
where the first line follows from the triangle inequality, and the second line uses \eqref{eq:approxProj}. Combining all of this, we have
\begin{align*}
\norm{\bx - \bx^{\ell+1}} &\le \norm{\bx - \widetilde{\bx}} + (1+\epsilon_2) \norm{\widetilde{\bx} - \cP_{\Gamma^*} \widetilde{\bx}} \\
&\le \norm{\bx - \widetilde{\bx}} + (1+\epsilon_2) \norm{\widetilde{\bx} - \bx} \\ &= (2+\epsilon_2) \norm{\bx - \widetilde{\bx}},
\end{align*}
where the second line follows from the fact that $\cP_{\Gamma^*} \widetilde{\bx}$ is the nearest neighbor to $\widetilde{\bx}$ among all vectors having a $k$-sparse representation in $\bD$.

\bibliographystyle{plain}
\footnotesize
\bibliography{preamble,cosampBib}

\end{document}